\documentclass[11pt]{article}
\usepackage[letterpaper]{geometry}
\usepackage{booktabs}
\usepackage{graphicx} % Required for inserting images
\usepackage{hyperref}
\usepackage{xcolor}
\usepackage{mathrsfs}
\usepackage{amsthm}
\usepackage{mathtools}
\usepackage{amssymb}
\usepackage{subcaption}
\usepackage{caption}
\usepackage{natbib}
\usepackage{geometry}
\usepackage{listings}
\usepackage{xparse}
\usepackage{setspace}
\usepackage{algorithmicx}
\usepackage{algpseudocode}
\usepackage{algorithm}
\usepackage{authblk}
\usepackage{mathabx}
\usepackage[english]{babel}
\newtheorem{theorem}{Theorem}

\theoremstyle{remark}

\newtheorem{lemma}{Lemma}

\title{Least Absolute Deviations Estimation for Sinusoidal Models}

\author{
    Zehaan Naik\thanks{Department of Statistical Science, 
    Duke University. 
    
    Email: znn@duke.edu}
    \and
    Debasis Kundu\thanks{Department of Mathematics and Statistics, 
    Indian Institute of Technology Kanpur. 
    
    Email: kundu@iitk.ac.in}
}
\date{\today}

\begin{document}
\onehalfspacing

\maketitle

\begin{abstract}
We study robust parameter estimation in sinusoidal regression models within a least absolute deviations (LAD) framework. While classical approaches rely predominantly on least-squares formulations, they are known to be sensitive to heavy-tailed noise and outliers. We formulate the estimation problem as direct minimization of the LAD objective and propose a simple, modular coordinate descent algorithm that exploits the partial convexity of the objective: amplitude parameters are updated via weighted median computations, leading to substantial computational improvements over traditional simplex-based optimization methods, while frequency parameters are estimated via a periodogram-inspired grid search with local refinement. We establish strong consistency and asymptotic normality of the proposed estimator under mild regularity conditions. Empirically, we demonstrate the method's effectiveness on both synthetic datasets and real-world time series, including the Mauna Loa atmospheric CO$_2$ data, air passenger data, and UK drivers’ deaths data, where robustness to non-Gaussian noise is essential. The proposed approach provides a simple, interpretable, and robust alternative to least-squares-based methods for sinusoidal signal estimation.

\end{abstract}

\section{Introduction}

Sinusoidal models form a fundamental building block in signal processing, time series analysis, and related domains, where oscillatory behavior arises naturally in applications ranging from communications and radar to climate dynamics and biomedical signals. A standard formulation considers observations of the form
\begin{equation}
y_t = A^0 \sin(\theta^0 t) + B^0 \cos(\theta^0 t) + \varepsilon_t, \quad t = 1, \dots, n,
\end{equation}
where $A, B \in \mathbb{R}$ denote amplitude parameters, $\theta \in (0, \pi)$ is the angular frequency, and $\{\varepsilon_t\}$ is a sequence of random errors. The precise regularity assumptions are introduced in Section~\ref{sec:background}; however, it is well known that assuming heavy-tailed distributions for the error terms yields a more robust model fit \citep{wang2024robustestimationhighdimensionaltime}.

Classical approaches to estimating the parameters of such models are predominantly based on Gaussian-noise assumptions and exploit least-squares, maximum-likelihood, or second-order statistical formulations, owing to the differentiability and tractability of the resulting objective functions and subspace methods \citep{Stoica2005, MUSIC, roy1989esprit, hua1990matrixpencil, li1996relax2, rife1974single, quinn1994estimating, kay1989fast}. However, LS estimators are known to be highly sensitive to deviations from Gaussian noise assumptions, particularly in the presence of heavy-tailed errors or outliers, settings that frequently arise in real-world signal data \citep{Huber1964}. This has led to increasing interest in robust estimation techniques, especially those based on least absolute deviations (LAD), which provide improved stability under non-ideal noise conditions. Robust sinusoidal parameter estimation has been studied in several contexts (see, e.g., \cite{kundu2024robust, kundu2, RenSinEstimation, SolakSinEstimation, AndoSinEstimation, DFTSinEst, tang2013atomic, katkovnik1998robust, li2010nonlinear}), though the primary focus has remained on least-squares formulations or asymptotic analysis under specific model assumptions.

Hence, in this work, we consider estimation of the parameters $(A^0, B^0, \theta^0)$ via direct minimization of the LAD objective
\begin{equation}
\label{EQ: Main-lad-objective}
Q_n(A, B, \theta) = \frac{1}{n} \sum_{t=1}^n \left| y_t - A \sin(\theta t) - B \cos(\theta t) \right|.
\end{equation}
While this objective offers robustness advantages, it presents significant computational challenges: it is convex in $(A, B)$ for fixed $\theta$, but highly non-convex in $\theta$, resulting in a difficult joint optimization problem. To the best of our knowledge, existing methods do not directly target this LAD formulation for sinusoidal regression, instead relying on least-squares-based periodogram methods or approximations that sacrifice robustness \citep{kundu2024robust, Stoica2005, LaplacePeriodogram}.

We propose a simple, modular algorithm for LAD-based estimation in sinusoidal models, based on an alternating minimization framework. The method exploits the partial convexity of the objective by iteratively updating the amplitude parameters $(A, B)$ via exact coordinate-wise LAD updates \citep{naik2026coordinatedescentalgorithmabsolute}, and the frequency parameter $\theta$ via direct minimization of the LAD objective. The latter is performed using a periodogram-inspired grid search combined with local refinement \citep{LaplacePeriodogram}, leveraging the model's oscillatory structure to efficiently explore the non-convex parameter space. This coordinate descent approach ensures monotone descent of the objective and is straightforward to implement \citep{TsengP}. We establish consistency and asymptotic normality of the proposed estimator under the LAD objective in Eq.~\eqref{EQ: Main-lad-objective}, building on and adapting the framework of \cite{roy2023consistencyasymptoticnormalityabsolute}.

Although our theoretical analysis focuses on the single-frequency model, the proposed framework extends naturally to multi-frequency settings of the form
\begin{equation}
y_t = \sum_{j=1}^p \left( A_j^0 \sin(\theta_j^0 t) + B_j^0 \cos(\theta_j^0 t) \right) + \varepsilon_t,
\end{equation}
where the same alternating structure can be applied component-wise. This highlights the approach's flexibility and its applicability to more complex signal structures. The theoretical results on consistency and asymptotic normality also extend to multiple frequency models, as discussed in Section~\ref{sec:background}.

We evaluate the proposed method on both synthetic datasets and real-world time series exhibiting different types of periodic behavior, including the Mauna Loa atmospheric CO$_2$ dataset \citep{Keeling1976}, the Air-Passengers dataset \citep{box2015time}, and the UK Driver-Deaths dataset \citep{harvey1990forecasting}. On synthetic data, we study performance as the number of components $p$ increases, including overparameterized regimes, and examine the role of initialization. We further conduct an outlier analysis on the Mauna Loa dataset by contaminating it with increasing proportions of heavy-tailed noise and compare the proposed LAD estimator with least-squares (LS) alternatives. These experiments show that the method remains stable as model complexity increases and is substantially more robust to outliers than LS-based approaches.

The remainder of the paper is organized as follows. Section~\ref{sec:background} reviews the univariate optimization results and the LAD reformulation of the periodogram idea that underpins coordinate-wise updates. Section~\ref{sec:methodology} introduces the proposed coordinate descent algorithm, establishes its convergence properties, and presents computational optimizations. Section~\ref{sec:experiments} reports empirical results on synthetic and real-world datasets, including high-dimensional and outlier-contaminated settings. Section~\ref{sec:conclusion} concludes with a discussion of practical implications and directions for future work.

\section{Background}
\label{sec:background}

In this section, we briefly review key properties of the least absolute deviations (LAD) objective that enable efficient optimization in the sinusoidal regression setting. In particular, we exploit the fact that, conditional on the frequency parameter, the model is linear in the amplitude parameters, allowing exact updates via median-based arguments. We then discuss how the classical periodogram intuition can be adapted to the LAD setting for frequency estimation.

\subsection{Univariate LAD Minimization}

We begin with two standard results for LAD minimization in univariate settings.

\begin{lemma}
Let $\{y_1, \dots, y_n\}$ be given. Then any minimizer of
\begin{equation}
\min_{a \in \mathbb{R}} \sum_{i=1}^{n} |y_i - a|
\end{equation}
is a sample median of $\{y_i\}$.
\end{lemma}

\begin{lemma}
Let $\{(x_i, y_i)\}_{i=1}^n$ be given with $x_i \neq 0$. Then any minimizer of
\begin{equation}
\min_{b \in \mathbb{R}} \sum_{i=1}^{n} |y_i - b x_i|
\end{equation}
is a weighted median of $\{y_i/x_i\}$ with weights $|x_i|$.
\end{lemma}

Both results follow from subgradient optimality conditions and are standard consequences of quantile regression theory \citep{koenker1978}.

\subsection{Application to the Sinusoidal Model}

Consider the sinusoidal model
\[
y_t = A^0 \sin(\theta^0 t) + B^0 \cos(\theta^0 t) + \varepsilon_t.
\]
Fixing $\theta$, the LAD objective in Eq.~\eqref{EQ: Main-lad-objective} reduces to
\[
\min_{A,B} \sum_{t=1}^n \left| y_t - A \sin(\theta t) - B \cos(\theta t) \right|.
\]
While the joint minimization over $(A,B)$ does not admit a closed form, coordinate-wise updates follow directly from the previous lemmas. Fixing $B$, define
\[
r_t^{(A)} = y_t - B \cos(\theta t).
\]
Then $A$ minimizes
\[
\sum_{t=1}^n \left| r_t^{(A)} - A \sin(\theta t) \right|,
\]
and is given by the weighted median of
\[
\left\{ \frac{r_t^{(A)}}{\sin(\theta t)} \right\}_{t=1}^n
\quad \text{with weights } |\sin(\theta t)|.
\]
Similarly, fixing $A$ and defining
\[
r_t^{(B)} = y_t - A \sin(\theta t),
\]
the update for $B$ is the weighted median of
\[
\left\{ \frac{r_t^{(B)}}{\cos(\theta t)} \right\}_{t=1}^n
\quad \text{with weights } |\cos(\theta t)|.
\]
These updates provide exact minimizers along each coordinate and form the basis of a coordinate descent procedure for $(A,B)$, conditional on $\theta$.

\subsection{LAD Reformulation of the Periodogram}

We now consider estimation of the frequency parameter $\theta$. For fixed $(A,B)$, define
\[
f_t(\theta) = A \sin(\theta t) + B \cos(\theta t).
\]
To understand the structure of the LAD objective, consider the discrepancy for some estimated $\hat\theta$:
\[
f_t(\theta) - f_t(\hat{\theta}) = A\big(\sin(\theta t) - \sin(\hat{\theta} t)\big) + B\big(\cos(\theta t) - \cos(\hat{\theta} t)\big).
\]
Using sum-to-product identities,
\[
\sin(\theta t) - \sin(\hat{\theta} t)
= 2 \cos\left(\frac{\theta + \hat{\theta}}{2} t\right)
  \sin\left(\frac{\theta - \hat{\theta}}{2} t\right),
\]
\[
\cos(\theta t) - \cos(\hat{\theta} t)
= -2 \sin\left(\frac{\theta + \hat{\theta}}{2} t\right)
   \sin\left(\frac{\theta - \hat{\theta}}{2} t\right).
\]
Substituting,
\[
f_t(\theta) - f_t(\hat{\theta})
= 2 \sin\left(\frac{\theta - \hat{\theta}}{2} t\right)
\underbrace{
\left[
A \cos\left(\frac{\theta + \hat{\theta}}{2} t\right)
- B \sin\left(\frac{\theta + \hat{\theta}}{2} t\right)
\right]}_{ = \varrho}.
\]
This representation separates the discrepancy into two components: a modulation term
\[
\sin\left(\frac{\theta - \hat{\theta}}{2} t\right),
\]
which vanishes systematically only when $\hat{\theta} = \theta$, and an oscillatory envelope involving $\frac{\theta + \hat{\theta}}{2}$.

The key observation is that the modulation term equals zero for all $t$ if and only if $\hat{\theta} = \theta$ (modulo periodic equivalence). For $\hat{\theta} \neq \theta$, it vanishes only on a sparse set of indices and remains nonzero for a set of $t$ with positive density. Since the underlined term ($\varrho$) is bounded and non-degenerate whenever $(A,B) \neq (0,0)$, it follows that
\[
\frac{1}{n} \sum_{t=1}^n \left| f_t(\theta) - f_t(\hat{\theta}) \right| > 0
\quad \text{for all } \hat{\theta} \neq \theta.
\]

This establishes identifiability of the frequency parameter under the LAD criterion. In particular, the LAD objective exhibits a periodogram-like structure, where global minima occur at the true frequency, while spurious cancellations remain local and do not persist across $t$.

Consequently, it is natural to estimate $\theta$ via a grid search over a dense set of candidate frequencies, followed by local refinement. The oscillatory structure ensures that incorrect frequencies incur a non-vanishing average discrepancy, allowing the true frequency to be isolated by directly evaluating the LAD objective.

\subsection{Consistency of the LAD Estimator}

We consider the model
\[
y_t = A^0 \sin(\theta^0 t) + B^0 \cos(\theta^0 t) + \varepsilon_t,
\qquad t = 1,\dots,n,
\]
where we assume that $(A^0,B^0) \in K \subset \mathbb{R}^2$ is compact, $\theta^0 \in (0,\pi)$,
and $\{\varepsilon_t\}$ are i.i.d.\ random variables with distribution function $F$
satisfying $F(0)=1/2$, continuous and bounded density $f$, and
$\mathbb{E}|\varepsilon_t|<\infty$.

Define the LAD objective
\[
Q_n(A,B,\theta)
=
\frac{1}{n}\sum_{t=1}^n
\bigl|y_t - A \sin(\theta t) - B \cos(\theta t)\bigr|,
\]
and let
\[
(\hat A_n,\hat B_n,\hat\theta_n)
=
\arg\min_{(A,B,\theta)\in\Theta}
Q_n(A,B,\theta),
\quad
\Theta = K \times (0,\pi).
\]

\begin{theorem}[Strong Consistency]
Under the stated assumptions,
\[
(\hat{A}_n, \hat{B}_n, \hat{\theta}_n)
\xrightarrow{\text{a.s.}}
(A^0, B^0, \theta^0).
\]
\end{theorem}

\begin{proof}

Presented in the Appendix

\end{proof}

\subsection{Asymptotic Normality of the LAD Estimator}

We consider the model
\[
y_t = A^0 \sin(\theta^0 t) + B^0 \cos(\theta^0 t) + \varepsilon_t,
\qquad t=1,\dots,n,
\]
where $(A^0,B^0) \in K \subset \mathbb{R}^2$ is compact, $\theta_0 \in (0,\pi)$, and
$\{\varepsilon_t\}$ are i.i.d.\ random variables with median zero, continuous
density $f$, and $f(0)>0$.

Define the residual
\[
r_t(A,B,\theta) = y_t - A \sin(\theta t) - B \cos(\theta t),
\]
and the LAD objective
\[
Q_n(A,B,\theta)
=
\frac{1}{n}\sum_{t=1}^n |r_t(A,B,\theta)|.
\]
Let $(\hat A_n,\hat B_n,\hat\theta_n)$ denote the LAD estimator, and suppose that
\[
(\hat A_n,\hat B_n,\hat\theta_n) \to (A^0,B^0,\theta^0)
\quad \text{almost surely}.
\]

\begin{theorem}[Asymptotic Normality]
Under the stated assumptions, and assuming
\[
(\hat A_n,\hat B_n,\hat\theta_n)
\to (A^0,B^0,\theta^0)
\quad \text{almost surely},
\]
we have
\[
\begin{pmatrix}
\sqrt n (\hat A_n-A^0) \\
\sqrt n (\hat B_n-B^0) \\
n^{3/2}(\hat\theta_n-\theta^0)
\end{pmatrix}
\xrightarrow{d}
\mathcal N\!\left(
0,\,
\frac{1}{4f(0)^2}
\begin{pmatrix}
2 & 0 & 0 \\
0 & 2 & 0 \\
0 & 0 & \dfrac{6}{(A^0)^2 + (B^0)^2}
\end{pmatrix}
\right).
\]
\end{theorem}

\begin{proof}
    Presented in the Appendix
\end{proof}

This result extends Theorem~4.2 of \cite{roy2023consistencyasymptoticnormalityabsolute}
to the sine-cosine parameterization of the single-frequency LAD regression model.

\subsection{Extension to Multiple-Frequency Models}
We briefly state the corresponding asymptotic results for the multi-frequency model
\[
y_t
=
\sum_{j=1}^p
\left(
A_j \sin(\theta_j t)
+
B_j \cos(\theta_j t)
\right)
+
\varepsilon_t,
\]
where the frequencies satisfy the identifiability condition
\[
\theta_j \neq \theta_k
\quad \text{for } j \neq k.
\]

Let
\[
\boldsymbol{\eta}
=
(A_1,B_1,\theta_1,\ldots,A_p,B_p,\theta_p)
\]
denote the parameter vector, and define the LAD objective
\[
Q_n(\boldsymbol{\eta})
=
\frac{1}{n}
\sum_{t=1}^n
\left|
y_t
-
\sum_{j=1}^p
\left(
A_j \sin(\theta_j t)
+
B_j \cos(\theta_j t)
\right)
\right|.
\]
Under the same assumptions as in the single-frequency setting, together with standard frequency-separation and identifiability conditions, the LAD estimator
\[
\hat{\boldsymbol{\eta}}_n
=
\arg\min_{\boldsymbol{\eta}} Q_n(\boldsymbol{\eta})
\]
is strongly consistent:
\[
\hat{\boldsymbol{\eta}}_n
\xrightarrow{\text{a.s.}}
\boldsymbol{\eta}_0.
\]

Furthermore, if the error density $f$ is continuous with $f(0)>0$, then the multi-frequency LAD estimator satisfies
\[
\sqrt{n}
\left(
\hat{A}_{1,n}-A_1,\,
\hat{B}_{1,n}-B_1,\,
\ldots,\,
\hat{A}_{p,n}-A_p,\,
\hat{B}_{p,n}-B_p
\right)
\xrightarrow{d}
N(0,\Sigma_A),
\]
while the frequency estimators satisfy
\[
n^{3/2}
\left(
\hat{\theta}_{1,n}-\theta_1,\,
\ldots,\,
\hat{\theta}_{p,n}-\theta_p
\right)
\xrightarrow{d}
N(0,\Sigma_\theta).
\]

Under standard frequency-separation conditions,
\[
\theta_i \neq \theta_j,
\qquad i\neq j,
\]
the trigonometric components corresponding to distinct frequencies become asymptotically orthogonal. In particular,
\[
\frac1n
\sum_{t=1}^n
\sin(\theta_i t)\sin(\theta_j t)
\to 0,
\qquad
\frac1n
\sum_{t=1}^n
\cos(\theta_i t)\cos(\theta_j t)
\to 0,
\]
and
\[
\frac1n
\sum_{t=1}^n
\sin(\theta_i t)\cos(\theta_j t)
\to 0,
\qquad i\neq j.
\]

Similarly, for the frequency-score terms,
\[
\frac1{n^3}
\sum_{t=1}^n
t^2
\cos(\theta_i t)\cos(\theta_j t)
\to 0,
\qquad i\neq j,
\]
with analogous relations for sine terms.

Consequently, the limiting information matrix assumes an asymptotically block-diagonal structure, implying that distinct sinusoidal components are asymptotically uncorrelated. The asymptotic covariance matrices are therefore determined component-wise by the corresponding single-frequency information matrices. The proofs follow by straightforward extension of the arguments used in the single-frequency case, together with standard trigonometric averaging results and nonlinear LAD M-estimation theory.

\section{Proposed Methodology}
\label{sec:methodology}

In this section, we present a coordinate descent algorithm for LAD estimation in sinusoidal models, along with its extension to multi-frequency settings. The proposed method exploits the partial convexity of the LAD objective: conditional on the frequency parameters, the problem is convex in the amplitude parameters and admits exact updates via weighted medians, while frequency estimation is handled via a direct minimization inspired by periodogram methods.

We begin with the single-frequency model and then extend the approach to the multi-frequency case. Practical considerations, including computational optimizations and initialization strategies, are discussed in subsequent subsections.

\subsection{Single-Frequency Model}

For the model
\[
y_t = A^0 \sin(\theta^0 t) + B^0 \cos(\theta^0 t) + \varepsilon_t,
\]
we seek to minimize the LAD objective in Eq.~\eqref{EQ: Main-lad-objective}. As discussed in Section~\ref{sec:background}, the objective is convex in $(A,B)$ for fixed $\theta$, but non-convex in $\theta$.

We therefore adopt an alternating minimization strategy. Given a current estimate of $\theta$, the amplitude parameters $(A,B)$ are updated via coordinate-wise LAD minimization using weighted medians. Given $(A,B)$, the frequency parameter $\theta$ is updated by minimizing the LAD objective over $\theta \in (0,\pi)$ using a grid search followed by local refinement.

\begin{algorithm}
\caption{LAD Estimation for Single-Frequency Sinusoidal Model}
\begin{algorithmic}[1]
\Require Observations $\{(t_i,y_i)\}_{i=1}^n$, initial values $(A^{(0)}, B^{(0)}, \theta^{(0)})$
\Ensure Estimated parameters $(\hat A, \hat B, \hat \theta)$
\State $k \gets 0$
\Repeat
\State $k \gets k + 1$
\State \textbf{Amplitude update:}
\State Update $A^{(k)}$ using weighted median of
\[
\left\{ \frac{y_t - B^{(k-1)}\cos(\theta^{(k-1)} t)}{\sin(\theta^{(k-1)} t)} \right\}
\]
with weights $|\sin(\theta^{(k-1)} t)|$
\State Update $B^{(k)}$ using weighted median of
\[
\left\{ \frac{y_t - A^{(k)}\sin(\theta^{(k-1)} t)}{\cos(\theta^{(k-1)} t)} \right\}
\]
with weights $|\cos(\theta^{(k-1)} t)|$
\State \textbf{Frequency update:}
\State Update $\theta^{(k)}$ by minimizing
\[
\sum_{t=1}^n \left| y_t - A^{(k)} \sin(\theta t) - B^{(k)} \cos(\theta t) \right|
\]
over $\theta \in (0,\pi)$ via grid search with local refinement
\Until convergence of $\theta^{(k)}$
\State \Return $(\hat A, \hat B, \hat \theta)$
\end{algorithmic}
\end{algorithm}

Each step minimizes the LAD objective along a coordinate block, ensuring monotone descent of the objective function.

\subsection{Extension to Multi-Frequency Models}

We now consider the model
\[
y_t = \sum_{j=1}^p \left( A_j \sin(\theta_j t) + B_j \cos(\theta_j t) \right) + \varepsilon_t.
\]
The corresponding LAD objective is
\[
\min_{\{A_j,B_j,\theta_j\}}
\frac{1}{n}\sum_{t=1}^n
\left|
y_t - \sum_{j=1}^p \left( A_j \sin(\theta_j t) + B_j \cos(\theta_j t) \right)
\right|.
\]

This objective remains convex in the amplitude parameters $\{A_j,B_j\}$ conditional on $\{\theta_j\}$, but is non-convex in the frequency parameters. We extend the alternating minimization framework by updating each component sequentially.

For each $j=1,\dots,p$, define the partial residual
\[
r_t^{(j)} = y_t - \sum_{\ell \neq j} \left( A_\ell \sin(\theta_\ell t) + B_\ell \cos(\theta_\ell t) \right).
\]

Conditional on $\{\theta_j\}$, the amplitude parameters are updated via weighted median steps:
\begin{itemize}
\item Update $A_j$ using $\{r_t^{(j)} / \sin(\theta_j t)\}$ with weights $|\sin(\theta_j t)|$
\item Update $B_j$ using $\{r_t^{(j)} / \cos(\theta_j t)\}$ with weights $|\cos(\theta_j t)|$
\end{itemize}

Conditional on $\{A_j,B_j\}$, each frequency $\theta_j$ is updated by minimizing
\[
\sum_{t=1}^n \left| r_t^{(j)} - A_j \sin(\theta t) - B_j \cos(\theta t) \right|
\]
over $\theta \in (0,\pi)$ using grid search with local refinement.

\begin{algorithm}[H]
\caption{LAD Estimation for Multi-Frequency Sinusoidal Model}
\begin{algorithmic}[1]
\Require Observations $\{(t_i,y_i)\}_{i=1}^n$, number of components $p$, initial values $\{A_j^{(0)},B_j^{(0)},\theta_j^{(0)}\}$
\Ensure Estimated parameters $\{A_j,B_j,\theta_j\}_{j=1}^p$
\State $k \gets 0$
\Repeat
\State $k \gets k + 1$
\For{$j = 1$ to $p$}
\State Compute partial residual $r_t^{(j)}$
\State Update $A_j^{(k)}$ via weighted median
\State Update $B_j^{(k)}$ via weighted median
\State Update $\theta_j^{(k)}$ via grid search with local refinement
\EndFor
\Until convergence
\State \Return $\{A_j,B_j,\theta_j\}_{j=1}^p$
\end{algorithmic}
\end{algorithm}

This extension preserves the structure of the single-frequency algorithm while enabling modeling of complex signals with multiple oscillatory components.

The non-convexity of the frequency parameters introduces sensitivity to initialization and computational challenges in practice. In the following subsections, we discuss efficient implementation strategies and principled initialization schemes to improve convergence and solution quality.

\subsection{Computational Considerations and Initialization}

The alternating minimization framework described above is straightforward to implement, but its efficiency and performance depend critically on computational optimizations and initialization strategies. We discuss both aspects here.

\paragraph{Global residuals for efficient updates.}
In the multi-frequency setting, recomputing partial residuals
\[
r_t^{(j)} = y_t - \sum_{\ell \neq j} \left( A_\ell \sin(\theta_\ell t) + B_\ell \cos(\theta_\ell t) \right)
\]
from scratch at each step incurs $\mathcal{O}(np)$ cost per coordinate update. To avoid this, we maintain a global fitted signal
\[
\hat y_t = \sum_{j=1}^p \left( A_j \sin(\theta_j t) + B_j \cos(\theta_j t) \right),
\]
and update residuals incrementally \citep{naik2026coordinatedescentalgorithmabsolute}. For each component $j$, we compute
\[
r_t^{(j)} = y_t - \hat y_t + \left( A_j \sin(\theta_j t) + B_j \cos(\theta_j t) \right),
\]
which allows reuse of previously computed quantities. After updating $(A_j,B_j,\theta_j)$, the global fit $\hat y_t$ is updated accordingly. This reduces the per-iteration complexity significantly and improves scalability for larger $p$.

\paragraph{Initialization of frequency parameters.}
The LAD objective is highly non-convex in the frequency parameters, and the algorithm may converge to different local minima depending on initialization. To ensure adequate coverage of the frequency domain, we initialize $\{\theta_j\}_{j=1}^p$ uniformly over $(0,\pi)$:
\[
\theta_j^{(0)} \approx \frac{j\pi}{p+1}, \qquad j=1,\dots,p.
\]
This choice spreads the initial frequencies across the admissible range and is consistent with standard discretization strategies used in spectral analysis and periodogram-based methods \citep{Stoica2005}.

In practice, it is often beneficial to choose $p$ larger than the anticipated number of dominant frequencies. The algorithm can then adapt by assigning negligible amplitudes to redundant components, effectively performing an implicit model selection. This overparameterization strategy improves robustness to initialization and reduces the risk of missing relevant frequency components.

\paragraph{Initialization of amplitude parameters.}
Given initial frequency estimates ${\theta_j}$, the amplitude parameters ${A_j,B_j}$ may be initialized either randomly or through a preliminary projection step. In particular, conditional on fixed frequencies, the sinusoidal model is linear in the amplitude parameters, allowing initial estimates to be obtained via least-squares projection onto the corresponding sinusoidal basis functions. Such initialization strategies are closely related to classical amplitude estimation methods in sinusoidal models, including least-squares and filter-bank based approaches, which are known to possess favorable statistical properties when the frequencies are reasonably well estimated \citep{StoicaAmpEstimation}. In our implementation, we initialize
\[
A_j^{(0)}, B_j^{(0)} \sim \mathrm{Unif}(\min_t y_t,\max_t y_t),
\]
which provides a simple data-adaptive starting point and avoids degenerate solutions. However, when computational resources permit, projection-based initializations obtained from least-squares or periodogram-derived frequency estimates may improve convergence and reduce sensitivity to local minima. The effects of initialization are examined empirically in Section~\ref{sec:experiments}.

\paragraph{Additional considerations.}
We note that several refinements can further improve performance:
\begin{itemize}
\item \textit{Multiple random restarts:} Running the algorithm from several independent initializations and selecting the solution with the lowest LAD objective helps mitigate local minima.
\item \textit{Adaptive frequency discretization:} Rather than performing optimization over a uniformly dense frequency grid, one may employ a hierarchical refinement strategy in which the LAD objective is first evaluated on a coarse discretization and then re-evaluated locally around candidate frequencies. Such adaptive grid selection reduces computational complexity and alleviates off-grid effects that arise in spectral estimation problems. The impact of discretization error and basis mismatch has been extensively studied in sparse spectral estimation and compressed sensing settings \citep{ChiGrid}, while adaptive refinement strategies have been proposed to mitigate these effects and improve frequency localization accuracy \citep{StoicaGridSelection}.
\item \textit{Component pruning:} Components with negligible amplitudes $(A_j^2 + B_j^2 \approx 0)$ can be removed during iterations to reduce model complexity.
\end{itemize}

These strategies are standard in non-convex optimization settings and are particularly effective for oscillatory models with multiple interacting components.

\section{Experiments}
\label{sec:experiments}
We evaluate the proposed LAD-based sinusoidal estimation method on both synthetic and real-world datasets. All experiments are implemented in \texttt{R}, and the code is made publicly available for reproducibility on our \href{https://github.com/Zehaan22/LAD_CD_sin_models}{GitHub repository}.

We consider three settings: (i) single-frequency synthetic data, (ii) multi-frequency synthetic data, and (iii) real-world environmental time series data.

\subsection{Synthetic Data Experiments}

We begin by evaluating the proposed LAD-based estimation procedure on synthetic datasets designed to systematically assess robustness and scalability. All experiments are conducted under controlled settings where the true data-generating process is known. We consider models of the form
\[
y_t = \sum_{j=1}^p \left( A_j^0 \sin(\theta_j^0 t) + B_j^0 \cos(\theta_j^0 t) \right) + \varepsilon_t,
\quad t = 1,\dots,n,
\]
with $n = 100$. The frequency parameters $\{\theta_j\}$ are sampled independently from $(0,\pi)$, and the amplitude parameters $\{A_j, B_j\}$ are drawn from a bounded interval. 

To evaluate robustness, we generate noise from a heavy-tailed distribution. In particular, we take
\[
\varepsilon_t \sim t_3,
\]
a Student-$t$ distribution with three degrees of freedom, which introduces substantial deviations from Gaussian assumptions and mimics realistic noise conditions encountered in practice.

We first consider the simplest case with a single sinusoidal component. In this setting, the proposed method reliably recovers the underlying signal structure, producing stable estimates of both amplitude and frequency parameters even in the presence of heavy-tailed noise. We next extend the analysis to multi-frequency models. For $p = 3$, the proposed coordinate descent algorithm successfully captures the superposition of oscillatory components. The use of alternating updates allows the method to decompose the signal into constituent frequencies. We compare the results with standard LSE fits to assess the effectiveness of the objective and descent strategy. We illustrate the fitted curves in Fig.~\ref{fig: syn_simple_models} and compare the final training mean absolute errors in Table~\ref{tab:syn_MAE_comp}. We can see that the LAD fits for the model compare well with the LSE fits and align closely with the true error added to the pure signal, indicating a reasonable fit. 

\begin{figure}
    \centering

    \begin{subfigure}{0.44\textwidth}
        \centering
        \includegraphics[width=\linewidth]{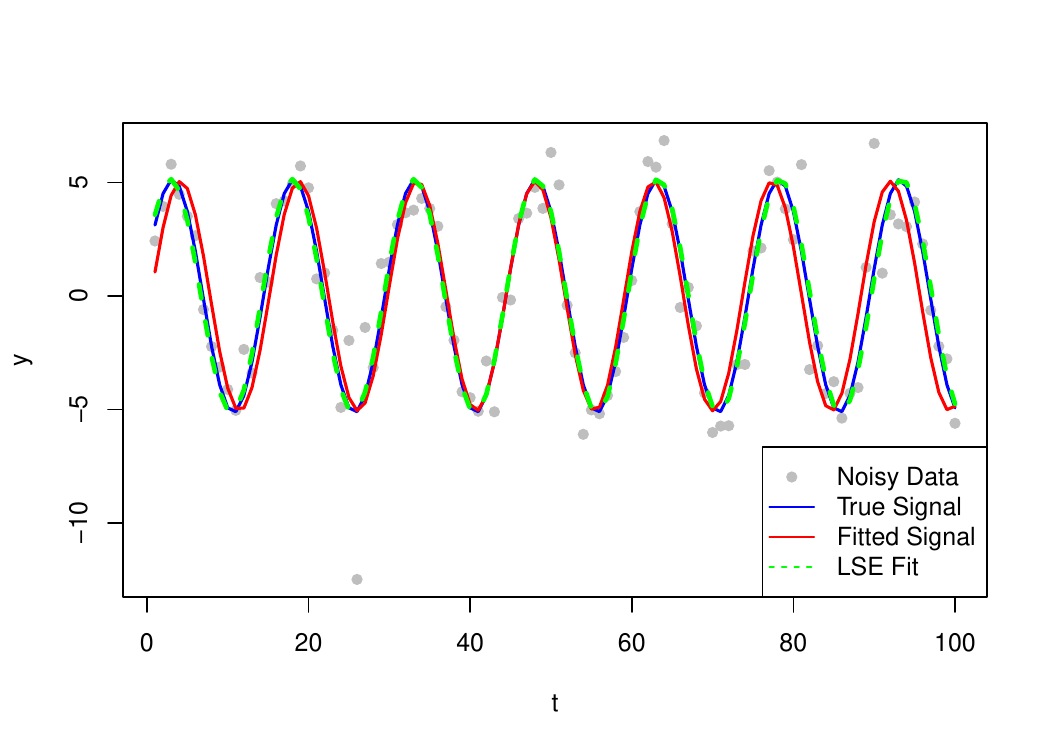}
        \caption{p = 1}
        \label{fig: syn_p1}
    \end{subfigure}
    \hfill
    \begin{subfigure}{0.51\textwidth}
        \centering
        \includegraphics[width=\linewidth]{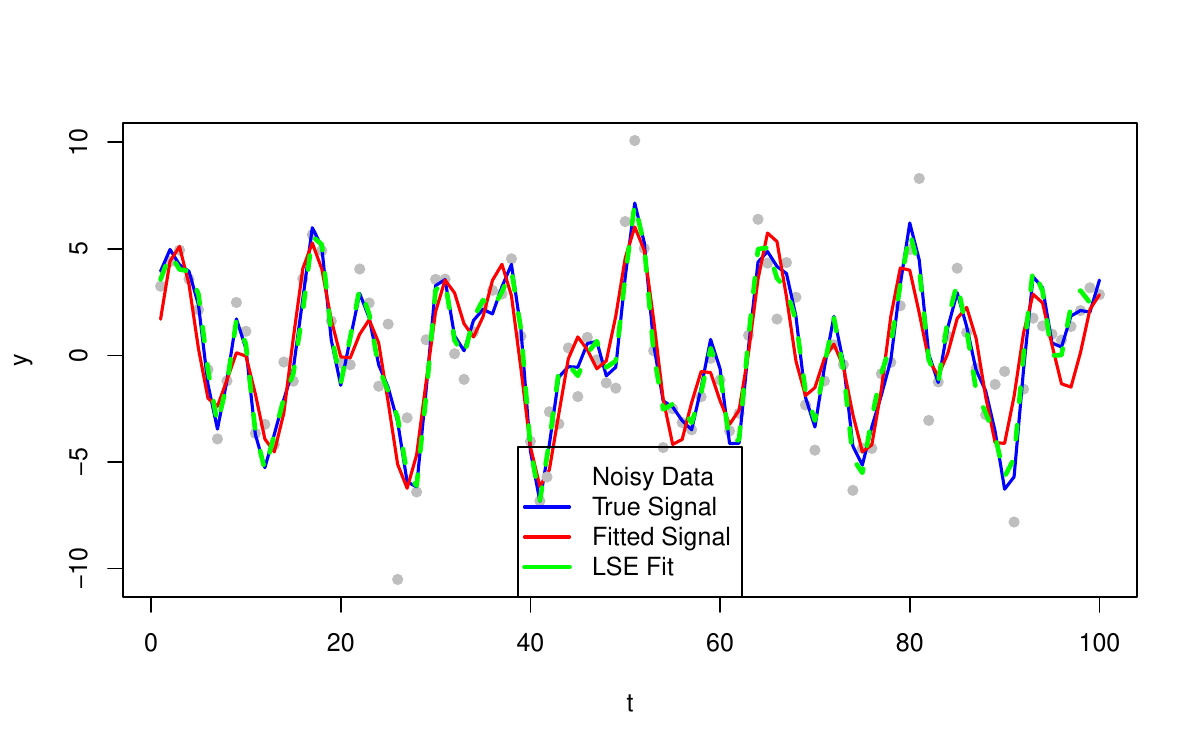}
        \caption{p = 3}
        \label{fig: syn_p2}
    \end{subfigure}

    \caption{Basic sinusoidal model fit comparison}
    \label{fig: syn_simple_models}
\end{figure}

\begin{table}
    \centering
    \begin{tabular}{|c|c|c|c|}
         \hline
         Model & Base MAE & MAE for LAD  & MAE for LSE\\
         \hline
         $p = 1$ & 1.054 & 1.342 & 1.083 \\
         $p = 3$ & 1.054 & 1.528 & 1.070 \\
         \hline
    \end{tabular}
    \caption{Comparing MAE for LSE and LAD model fits}
    \label{tab:syn_MAE_comp}
\end{table}

\paragraph{Scaling with model complexity.}
To further investigate the behavior of the method in high-dimensional regimes, we fix $n = 500$ and vary the number of components $p$ over a wide range:
\[
p \in \{1, 5, 10, 20, 50, 100, 200, 500, 1000\}.
\]
For each value of $p$, synthetic data are generated according to the model above, and the LAD estimator is computed. We report performance using the mean absolute deviation (MAD) between the observed data and both the true and fitted signals:
\[
\mathrm{MAD} = \frac{1}{n}\sum_{t=1}^n |y_t - \hat y_t|.
\]

\paragraph{Effect of initialization.}
Given the non-convexity of the optimization problem, initialization plays a critical role in determining the quality of the final solution. To study this effect, we consider multiple initialization strategies. In addition to a baseline random initialization, we construct warm starts by perturbing the true parameters with Gaussian noise of varying standard deviation:
\[
\theta_j^{(0)} = \theta_j + \eta_j, \quad \eta_j \sim \mathcal{N}(0, \sigma^2),
\]
with analogous perturbations applied to $A_j$ and $B_j$. We consider $\sigma \in \{0.1, 1, 2, \sqrt{10}\}$ to span a range from near-oracle initialization to highly noisy starting points. Table~\ref{tab:synthetic_warmstart} reports the resulting MAD values across different values of $p$ and initialization strategies.

Several trends emerge from these experiments. First, the baseline random initialization becomes increasingly unstable as $p$ increases, and the fitted MAD deteriorates rapidly in overparameterized regimes. In contrast, warm-start strategies significantly improve performance, particularly when the initialization is close to the true parameters. Even moderate perturbations (e.g., $\sigma = 1$) yield substantial gains over random initialization. In extreme cases, poor initialization leads to divergence or highly suboptimal fits (due to the non-convexity of the objective). Overall, the proposed method demonstrates strong robustness to heavy-tailed noise and the ability to scale to high-dimensional settings, provided that reasonable initialization strategies are employed.

\begin{table}
\centering
\caption{Effect of initialization quality on LAD estimation across model complexity}
\begin{tabular}{|r|c|c|c|c|c|c|}
\hline
$p$ & MAD (True) & Random Init & sd = 0.1 & sd = 1 & sd = 2 & sd = $\sqrt{10}$ \\
\hline
1    & 1.0783 & 1.1043 & 2.1460 & 1.1043 & 1.1043 & 1.1043 \\
5    & 1.0689 & 1.4663 & 1.5057 & 1.5032 & 1.5051 & $6.0783 \times 10^6$ \\
10   & 1.1543 & 3.6558 & 2.6125 & 2.8213 & $1.9623 \times 10^7$ & $6.4517 \times 10^7$ \\
20   & 1.1200 & 5.5499 & 3.5503 & $1.4008 \times 10^6$ & $5.6659 \times 10^5$ & $2.0304 \times 10^7$ \\
50   & 1.0978 & 8.7855 & 4.9486 & $2.9989 \times 10^5$ & $4.9118 \times 10^7$ & $1.2251 \times 10^7$ \\
100  & 1.0709 & 11.4082 & 6.1202 & $3.9978 \times 10^4$ & $1.8113 \times 10^5$ & $3.3632 \times 10^6$ \\
200  & 1.0499 & 117.8261 & 6.2874 & $1.5842 \times 10^4$ & $1.1499 \times 10^6$ & $6.7248 \times 10^6$ \\
500  & 1.0686 & $3.8731 \times 10^3$ & 11.4654 & $1.6691 \times 10^4$ & $2.0627 \times 10^5$ & $1.4869 \times 10^6$ \\
1000 & 1.0288 & $1.2277 \times 10^4$ & 14.8495 & $5.0821 \times 10^3$ & $4.0419 \times 10^3$ & $1.2801 \times 10^4$ \\
\hline
\end{tabular}
\label{tab:synthetic_warmstart}
\end{table}

\subsection{Real Data: Mauna Loa CO$_2$ Analysis}

We evaluate the proposed LAD-based sinusoidal estimation framework on the Mauna Loa atmospheric CO$_2$ dataset \citep{Keeling1976}, a canonical example of a real-world time series exhibiting both periodic structure and non-Gaussian noise characteristics. The data are first detrended to isolate the oscillatory component, and the model is fit using the proposed algorithm under different configurations.

We begin by applying the proposed sinusoidal model with both sine and cosine components directly to the data for varying values of $p$. Figure~\ref{fig:mauna_naive} shows the resulting fits for $p \in \{5,10,20,50\}$. 

\begin{figure}
    \centering

    \begin{subfigure}{0.45\textwidth}
        \centering
        \includegraphics[width=\linewidth]{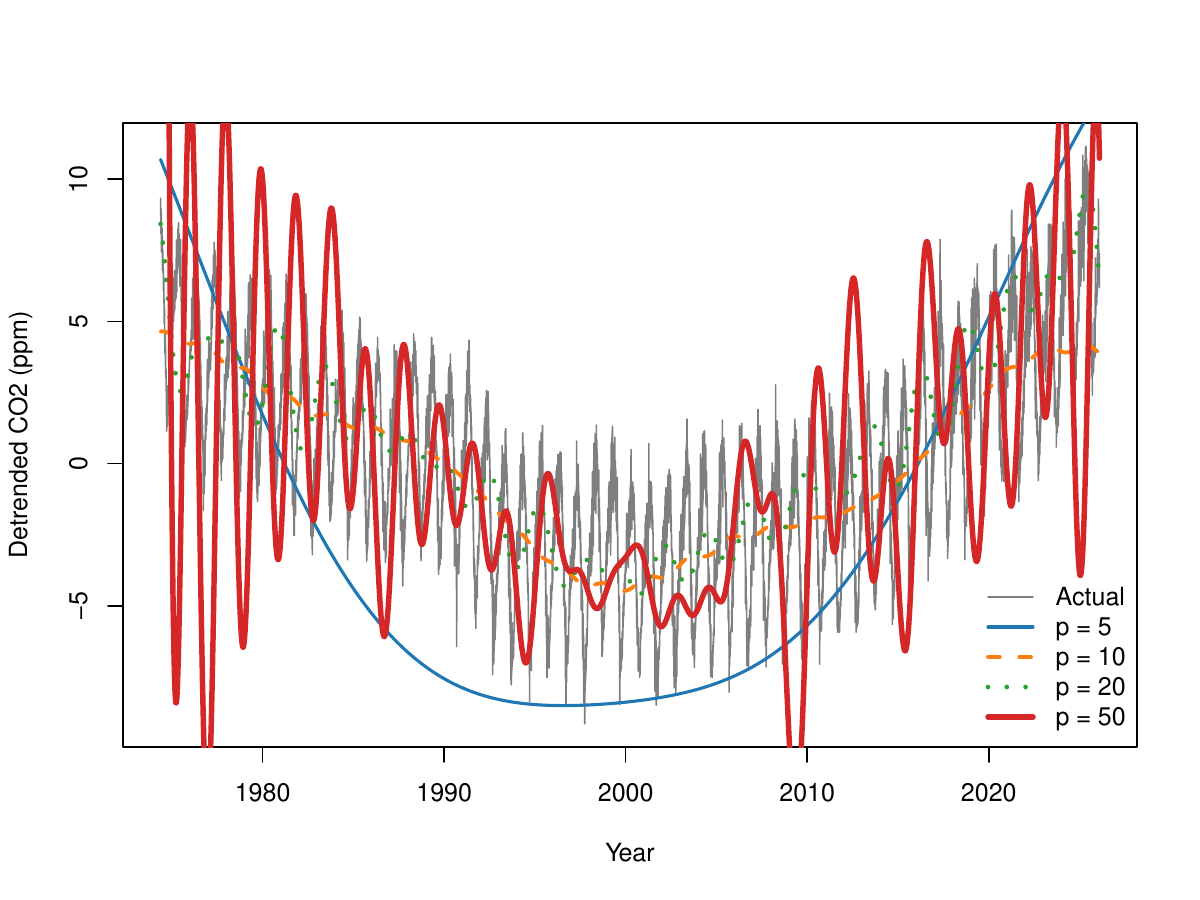}
        \caption{Naive model fit}
        \label{fig:mauna_naive}
    \end{subfigure}
    \hfill
    \begin{subfigure}{0.45\textwidth}
        \centering
        \includegraphics[width=\linewidth]{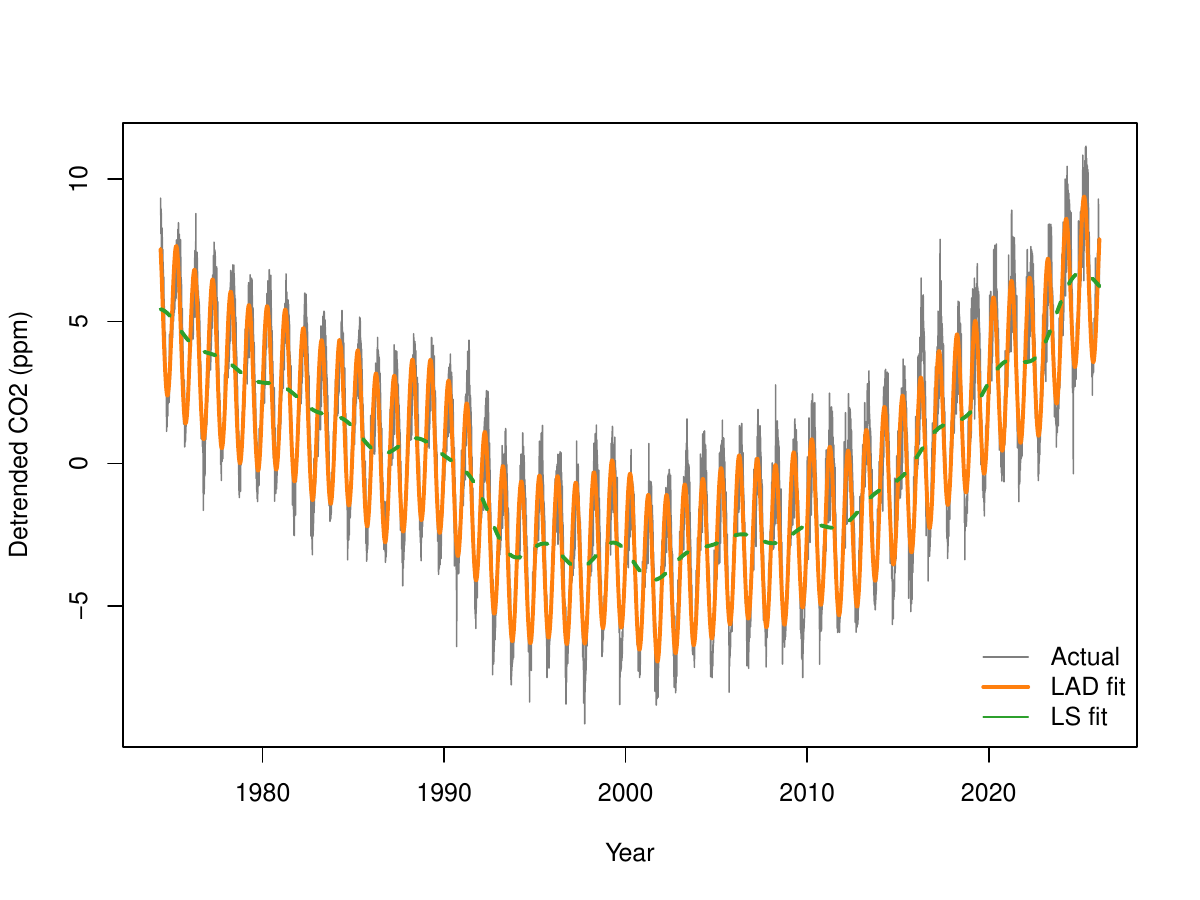}
        \caption{Only sin model fit compared to LSE}
        \label{fig:mauna_sin_only}
    \end{subfigure}

    \caption{LAD model fits for Mauana Lao CO$_2$ data}
    \label{fig: Mauana_fit}
\end{figure}

As seen in Figure~\ref{fig:mauna_naive}, the naive application of the coordinate descent algorithm does not yield satisfactory fits. Increasing the number of components leads to instability and overfitting behavior, with the model failing to capture the underlying seasonal structure in a coherent manner. This highlights the sensitivity of the full sine-cosine parameterization to initialization and optimization challenges in real-world settings.

To address this, we consider a simplified formulation using only sine components, effectively fixing the phase structure and reducing the complexity of the optimization problem. This modification leads to significantly improved fits. Figure~\ref{fig:mauna_sin_only} presents the resulting LAD fit alongside the least-squares (LS) fit for $p=10$.

As illustrated in Figure~\ref{fig:mauna_sin_only}, the LAD-based fit closely tracks the observed oscillatory behavior while maintaining stability across the time horizon. The LS fit performs comparably in smooth regions but is more sensitive to local fluctuations (MAD for the LAD fit is 0.896 and that for LSE is 1.921). This experiment demonstrates that modest structural modifications to the model—such as restricting to sine components—can substantially improve optimization performance while retaining flexibility. More broadly, it highlights the adaptability of the proposed framework to different modeling choices.

Finally, we assess robustness by introducing synthetic outliers into the dataset. We contaminate the observations by adding large-magnitude noise to a randomly selected subset of points, with contamination levels ranging from $5\%$ to $30\%$.

Figure~\ref{fig:mauna_outliers} shows the resulting fits under increasing levels of contamination.

\begin{figure}
    \centering
    \includegraphics[width=0.9\linewidth]{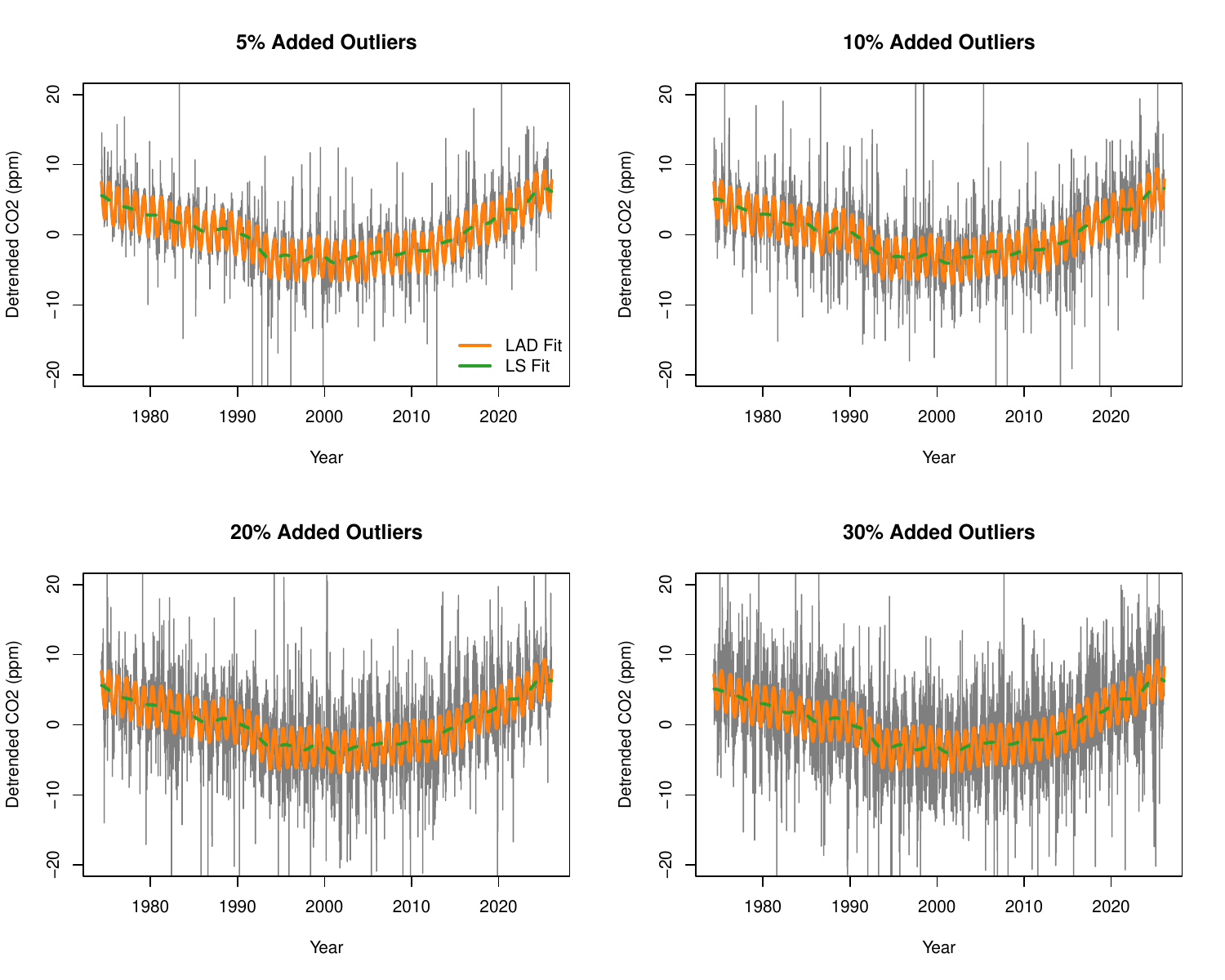}
    \caption{Outlier analysis for LAD and LSE fits}
    \label{fig:mauna_outliers}
\end{figure}

As seen in Figure~\ref{fig:mauna_outliers}, the LAD-based estimator remains stable even under substantial contamination, preserving the underlying seasonal structure. In contrast, the LS-based fit deteriorates rapidly as the proportion of outliers increases, exhibiting large deviations and loss of structure.

We further quantify this behavior using two metrics: (i) mean absolute error (MAE) between the observed (without added outliers) and fitted signals, and (ii) parameter recovery error, defined as the mean absolute deviation between estimated parameters for all contaminated datasets and reference parameters (taken to be the estimates from the original observed data). The numerical results, reported in Table~\ref{tab:outlier_analysis}, confirm that the LAD estimator consistently outperforms LS under increasing contamination.

\begin{table}
\centering
\caption{Performance under increasing outlier contamination}
\begin{tabular}{|c|c|c|c|c|}
\hline
Outliers (\%) & MAE (LAD) & MAE (LS) & MAE for $A$s (LAD) & MAE for $\theta$s (LAD) \\
\hline
0   &  0.8957 & 1.9206 & 0.0000 & 0.0000 \\
5   &  0.9094 & 1.9218 & 0.1878 & 0.4207 \\
10  &  0.8954 & 1.9209 & 0.1684 & 0.4523 \\
20  &  0.9028 & 1.9230 & 0.2985 & 0.2822 \\
30  &  0.9114 & 1.9226 & 0.3066 & 0.2980 \\
\hline
\end{tabular}
\label{tab:outlier_analysis}
\end{table}

These experiments illustrate three key points. First, while the full sine-cosine parameterization is theoretically appealing, it can be difficult to optimize directly in complex real-world settings. Second, appropriate structural modifications, such as simplifying the parameterization, can significantly improve performance without sacrificing interpretability. Third, the LAD objective provides strong robustness to outliers, making it particularly well-suited for noisy environmental time series data.

\subsection{AirPassengers Data}

We next consider the AirPassengers dataset, a classical time series exhibiting strong seasonal structure \citep{box2015time}. The data consist of monthly airline passenger counts over a relatively short time horizon ($n \approx 140$), making it well-suited for studying the behavior of the proposed method under limited sample sizes. We fit the proposed multi-frequency LAD model for increasing values of $p \in \{3, 5, 10 \}$. Figure~\ref{fig:air_fit} shows the resulting fits.

\begin{figure}[htbp]
\centering

\begin{minipage}{0.55\textwidth}
\centering
\includegraphics[width=\textwidth]{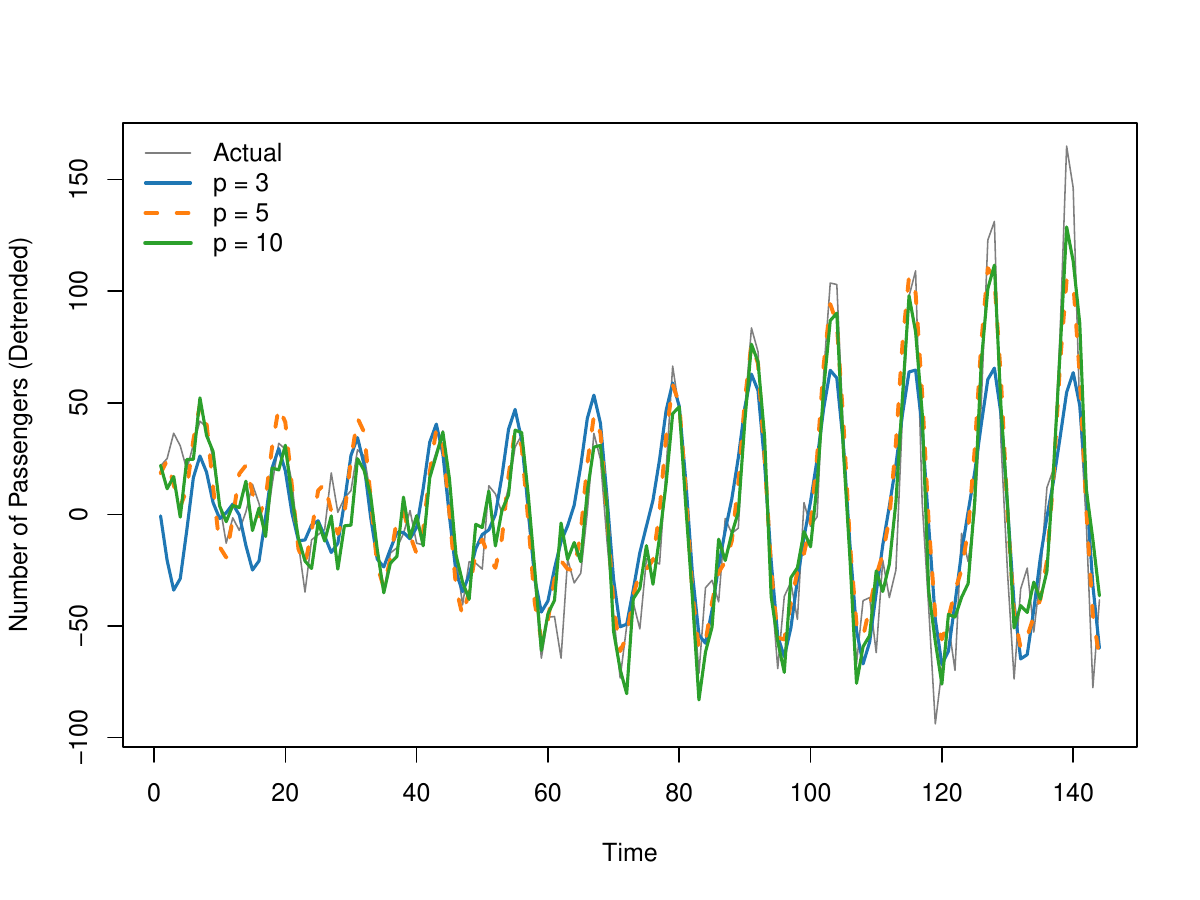}

\captionof{figure}{LAD fits for increasing $p$}
\label{fig:air_fit}
\end{minipage}
\hfill
\begin{minipage}{0.4\textwidth}
\centering

    \begin{tabular}{|c|c|}
        \hline
        $p$ & MAE\\
        \hline
        3  & 19.9252 \\
        5  & 14.3049 \\
        10 & 12.7891 \\
        20 & 511.0688 \\
\hline
\end{tabular}

\captionof{table}{Performance metrics.}
\label{tab:air_results}
\end{minipage}

\end{figure}

As seen in Figure~\ref{fig:air_fit}, the model progressively improves as $p$ increases. For small values of $p$, the fit captures the dominant seasonal pattern, whereas larger values allow the model to capture finer variations in the data. In particular, $p=10$ provides a good balance between flexibility and stability. We summarize the numerical performance of the model in Table~\ref{tab:air_results}, reporting mean absolute error (MAE) for different values of $p$.

Due to the relatively small sample size, increasing $p$ beyond this range leads to overfitting. In particular, for $p=20$, the model becomes unstable and exhibits erratic behavior, reflecting the difficulty of estimating a large number of parameters from limited data. This highlights a key trade-off between model complexity and sample size in multi-frequency settings. This experiment highlights the proposed method's ability to capture structured seasonal patterns with relatively few components, while also illustrating the risks of overparameterization in small-sample settings.

\subsection{UK Driver Deaths Data}

We next consider the UKDriverDeaths dataset, which records monthly counts of road accident fatalities in the United Kingdom and exhibits a clear seasonal structure with moderate noise \citep{harvey1990forecasting}. Compared to the AirPassengers data, this dataset presents a slightly more realistic setting with increased variability while still retaining strong periodic patterns. We fit the proposed multi-frequency LAD model for increasing values of $p \in \{3, 5, 10\}$. Figure~\ref{fig:uk_fit} shows the resulting fits.

\begin{figure}
\centering

\begin{minipage}{0.55\textwidth}
\centering
\includegraphics[width=\textwidth]{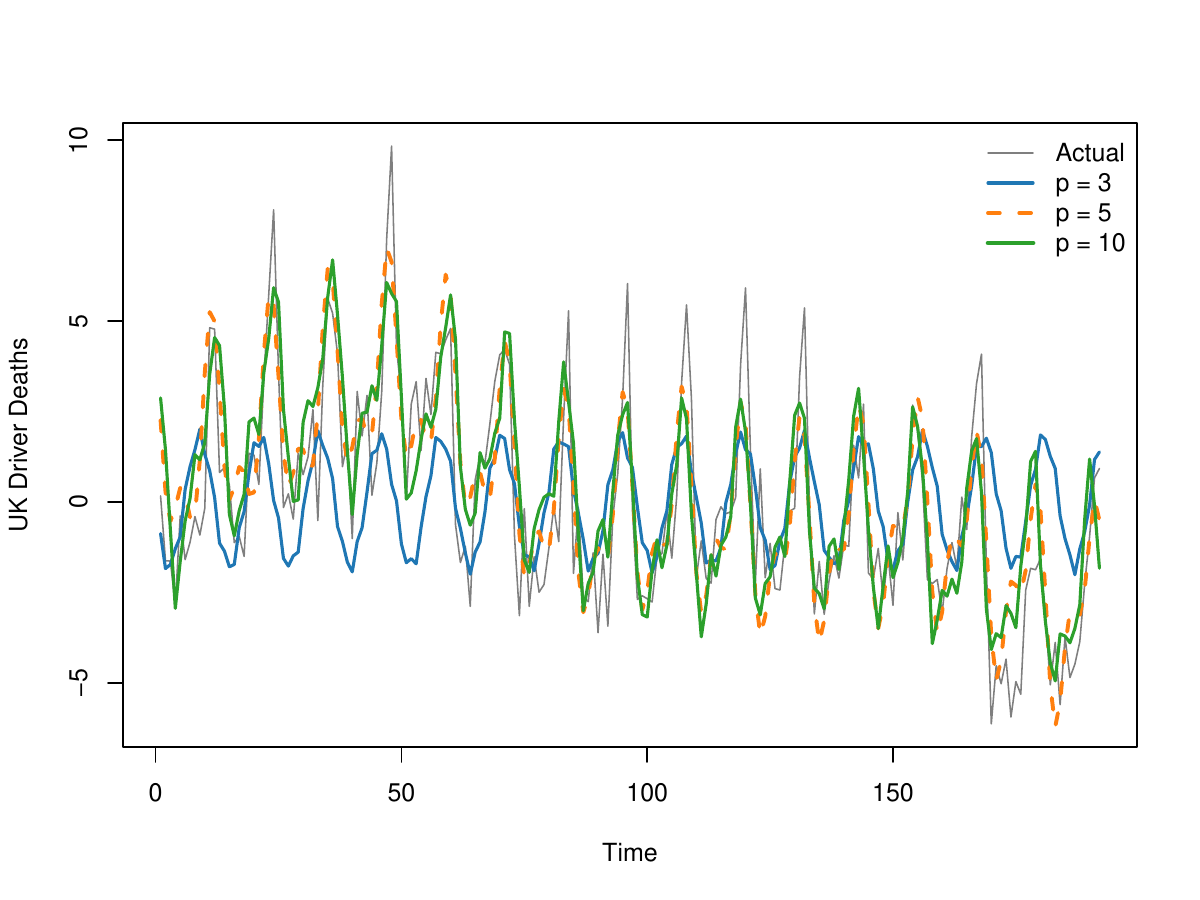}

\captionof{figure}{LAD fits for increasing $p$}
\label{fig:uk_fit}
\end{minipage}
\hfill
\begin{minipage}{0.4\textwidth}
\centering

\begin{tabular}{|c|c|}
\hline
$p$ & MAE \\
\hline
3  & 2.0423 \\
5  & 1.1823 \\
10 & 1.3266 \\
20 & 7.1713 \\
\hline
\end{tabular}

\captionof{table}{Performance metrics.}
\label{tab:uk_results}
\end{minipage}

\end{figure}

As seen in Figure~\ref{fig:uk_fit}, the model successfully captures the dominant seasonal structure of the data even for small values of $p$. Increasing $p$ from $3$ to $5$ leads to a substantial improvement in fit, as reflected in the reduction in MAE reported in Table~\ref{tab:uk_results}. However, unlike the AirPassengers dataset, further increasing the model complexity does not yield consistent improvements. In particular, the fit for $p=10$ shows only marginal degradation in MAE, whereas for $p=20$ the model becomes unstable and exhibits clear overfitting. This is also visible in the fitted curves, which begin to capture noise rather than the underlying seasonal structure. In practice, relatively small values of $p$ are sufficient to capture the series' essential dynamics, and increasing $p$ beyond this range yields diminishing returns and potential instability.

\section{Conclusion}
\label{sec:conclusion}

In this work, we developed a robust framework for parameter estimation in sinusoidal regression models based on the least absolute deviations (LAD) objective. By directly optimizing the LAD criterion, we address the well-known sensitivity of least-squares methods to heavy-tailed noise and outliers. 

The proposed approach leverages the partial convexity of the objective through a coordinate descent scheme: amplitude parameters are updated via exact weighted median computations, while frequency parameters are estimated using a periodogram-inspired search strategy. This results in a simple, modular algorithm that is both interpretable and easy to implement. We established strong consistency and asymptotic normality of the estimator under mild conditions, extending existing results to the sine-cosine formulation of the model.

Through extensive experiments on synthetic and real-world datasets, we demonstrated several key properties of the method. The estimator remains stable as model complexity increases when appropriately initialized, exhibits strong robustness to outliers compared to least-squares methods, and adapts effectively to different types of periodic signals. At the same time, our results highlight important practical considerations, including sensitivity to initialization (and several initialization techniques) and the challenges posed by overparameterized regimes.

Several directions for future work remain. A natural extension is to incorporate regularization into the LAD framework, particularly in multi-frequency settings. Penalization schemes for amplitude parameters or the number of active components can improve stability, enable automatic model selection, and mitigate overfitting in high-dimensional regimes. Additionally, more structured parameterizations and adaptive frequency-selection strategies may further improve performance on complex real-world signals.

Overall, the proposed LAD-based approach provides a flexible and robust alternative to classical least-squares methods for sinusoidal estimation, with both strong theoretical guarantees and practical applicability across a range of signal processing problems.

\newpage
\bibliography{references}
\bibliographystyle{apalike}

\newpage
\section*{Appendix}

\setcounter{theorem}{0}
\begin{theorem}[Strong Consistency]
Under the stated assumptions,
\[
(\hat{A}_n, \hat{B}_n, \hat{\theta}_n)
\xrightarrow{\text{a.s.}}
(A^0, B^0, \theta^0).
\]
\end{theorem}

\begin{proof}

We define the contrast function
\[
H_n(A,B,\theta)
=
Q_n(A,B,\theta) - Q_n(A^0,B^0,\theta^0),
\]
so that $H_n(A^0,B^0,\theta^0)=0$. Writing
\[
h_t(A,B,\theta)
=
A^0 \sin(\theta^0 t) + B^0 \cos(\theta^0 t)
- A \sin(\theta t) - B \cos(\theta t),
\]
we obtain
\begin{align*}
H_n(A,B,\theta)
&=
\frac{1}{n}\sum_{t=1}^n
\Bigl(
\bigl|y_t - A\sin(\theta t) - B\cos(\theta t)\bigr|
-
\bigl|y_t - A^0\sin(\theta^0 t) - B^0\cos(\theta^0 t)\bigr|
\Bigr) \\
&=
\frac{1}{n}\sum_{t=1}^n
\Bigl(
\bigl|\varepsilon_t + h_t(A,B,\theta)\bigr|
-
\bigl|\varepsilon_t\bigr|
\Bigr).
\end{align*}

\begin{lemma}
    \label{Lem: error_inequality}
    For any real number $h$,
    \[
    \mathbb{E}\bigl[|\varepsilon + h| - |\varepsilon|\bigr] \ge 0,
    \]
    with equality if and only if $h=0$.

    \begin{proof}
        Since $\varepsilon$ has median zero,
        \[
        \mathbb{E}|\varepsilon + h| - \mathbb{E}|\varepsilon|
        =
        2\int_0^{|h|}
        \bigl(F(u)-\tfrac12\bigr)\,du
        \ge 0,
        \]
        because $F(u)>\tfrac12$ for all $u>0$.
    \end{proof}

\end{lemma}

As a consequence, for each fixed $(A,B,\theta)$,
\[
\mathbb{E}\bigl[H_n(A,B,\theta)\bigr]
=
\frac{1}{n}\sum_{t=1}^n
\mathbb{E}\Bigl[
|\varepsilon_t + h_t(A,B,\theta)| - |\varepsilon_t|
\Bigr]
\ge 0.
\]
Moreover, by Lemma~\ref{Lem: error_inequality}, there exists $c>0$ such that
\[
\mathbb{E}\Bigl[
|\varepsilon_t + h_t(A,B,\theta)| - |\varepsilon_t|
\Bigr]
\ge
c\,|h_t(A,B,\theta)|.
\]

\paragraph{Uniform convergence.}
Let
\[
Z_t(A,B,\theta)
=
|\varepsilon_t + h_t(A,B,\theta)| - |\varepsilon_t|.
\]
Then
\[
|Z_t(A,B,\theta)|
\le
2\bigl(|\varepsilon_t| + |A^0| + |B^0| + |A| + |B|\bigr),
\]
and since $K$ is compact and $\mathbb{E}|\varepsilon_t|<\infty$,
$\mathbb{E}|Z_t(A,B,\theta)|$ is uniformly bounded over $\Theta$.
Furthermore, $Z_t(A,B,\theta)$ is continuous in $(A,B,\theta)$.

Since $\Theta$ is compact, $Z_t(A,B,\theta)$ is continuous in 
$(A,B,\theta)$ for every $t$, and
\[
|Z_t(A,B,\theta)|
\leq
2\left(
|\varepsilon_t|
+ |A_0| + |B_0| + |A| + |B|
\right),
\]
where the dominating function has a finite expectation uniformly over 
$\Theta$, the class
\[
\{Z_t(A,B,\theta):(A,B,\theta)\in\Theta\}
\]
is Glivenko--Cantelli. Hence, by a standard Uniform Law of Large Numbers 
for compact parametric families (see, \cite{newey1994large}),
\[
\sup_{(A,B,\theta)\in\Theta}
\left|
H_n(A,B,\theta)-EH_n(A,B,\theta)
\right|
\xrightarrow{\text{a.s.}}0.
\]

Hence,
\[
H_n(A,B,\theta)
\xrightarrow{\text{a.s.}}
Q(A,B,\theta)
\equiv
\lim_{n\to\infty}
\mathbb{E}H_n(A,B,\theta),
\]
uniformly on $\Theta$.

\paragraph{Identification of the limit.}
From the previous inequality,
\[
Q(A,B,\theta)
\ge
c
\lim_{n\to\infty}
\frac{1}{n}\sum_{t=1}^n
|h_t(A,B,\theta)|.
\]
Thus it suffices to show that
\[
I(A,B,\theta)
\equiv
\lim_{n\to\infty}
\frac{1}{n}\sum_{t=1}^n
\bigl|
h_t(A,B,\theta)
\bigr|
> 0
\quad
\text{for }(A,B,\theta)\neq(A^0,B^0,\theta^0).
\]

\paragraph{Case 1: $\theta=\theta_0$, $(A,B)\neq(A^0,B^0)$.}
Then
\[
h_t = (A^0-A)\sin(\theta^0 t) + (B^0-B)\cos(\theta^0 t).
\]
This is a non-trivial sinusoidal signal. Using standard averaging arguments,
\[
\lim_{n\to\infty}
\frac{1}{n}\sum_{t=1}^n |h_t|
= c_1 \sqrt{(A^0-A)^2 + (B^0-B)^2} > 0,
\]
for some constant $c_1>0$.

\paragraph{Case 2: $\theta\neq\theta_0$.}
Write
\[
h_t =
(A^0 \sin(\theta^0 t) + B^0 \cos(\theta^0 t))
-
(A \sin(\theta t) + B \cos(\theta t)).
\]
Using the trigonometric expansions from the previous subsection,
this can be expressed as a sum of terms involving
\[
\sin\!\left(\tfrac{\theta_0-\theta}{2}t\right)
\quad \text{and oscillatory factors}.
\]
The modulation term vanishes only on a sparse set of $t$, while the
remaining factors are bounded and non-degenerate whenever $(A^0,B^0)\neq(0,0)$.
Hence
\[
\frac{1}{n}\sum_{t=1}^n |h_t| > 0
\quad \text{for all } \theta \neq \theta^0.
\]

Therefore,
\[
Q(A,B,\theta)=0
\quad\text{if and only if}\quad
(A,B,\theta)=(A^0,B^0,\theta^0).
\]

\paragraph{Conclusion.}
Since $H_n(A,B,\theta)$ converges uniformly almost surely to $Q(A,B,\theta)$,
and $Q(A,B,\theta)$ has a unique global minimum at $(A^0,B^0,\theta^0)$,
it follows by standard M-estimation theory that
\[
(\hat A_n,\hat B_n,\hat\theta_n)
\xrightarrow{\text{a.s.}}
(A^0,B^0,\theta^0).
\]

\end{proof}

\begin{theorem}[Asymptotic Normality]
Under the stated assumptions, and assuming
\[
(\hat A_n,\hat B_n,\hat\theta_n)
\to (A^0,B^0,\theta^0)
\quad \text{almost surely},
\]
we have
\[
\begin{pmatrix}
\sqrt n (\hat A_n-A^0) \\
\sqrt n (\hat B_n-B^0) \\
n^{3/2}(\hat\theta_n-\theta^0)
\end{pmatrix}
\xrightarrow{d}
\mathcal N\!\left(
0,\,
\frac{1}{4f(0)^2}
\begin{pmatrix}
2 & 0 & 0 \\
0 & 2 & 0 \\
0 & 0 & \dfrac{6}{(A^0)^2 + (B^0)^2}
\end{pmatrix}
\right).
\]
\end{theorem}

\begin{proof}
Since the subgradient of $|x|$ is $\operatorname{sgn}(x)$ almost surely, the estimating equations are
\begin{align*}
0
&=
\frac1n
\sum_{t=1}^n
\operatorname{sgn}(r_t(A,B,\theta))
\sin(\theta t),
\\
0
&=
\frac1n
\sum_{t=1}^n
\operatorname{sgn}(r_t(A,B,\theta))
\cos(\theta t),
\\
0
&=
\frac1n
\sum_{t=1}^n
\operatorname{sgn}(r_t(A,B,\theta))
\Bigl(
A t\cos(\theta t)
-
B t\sin(\theta t)
\Bigr).
\end{align*}

At the true parameter,
\[
r_t(A^0,B^0,\theta^0)=\varepsilon_t.
\]

We use the convention
\[
\operatorname{sgn}(x)=
\begin{cases}
+1, & x>0,\\
0, & x=0,\\
-1, & x<0.
\end{cases}
\]

\paragraph{Local parameterization.}

Let
\[
u=
\begin{pmatrix}
u_1\\
u_2\\
u_3
\end{pmatrix},
\qquad
A=A^0+\frac{u_1}{\sqrt n},
\qquad
B=B^0+\frac{u_2}{\sqrt n},
\qquad
\theta=\theta^0+\frac{u_3}{n^{3/2}}.
\]

We study the local behavior of
\[
Q_n\!\left(
A^0+\frac{u_1}{\sqrt n},
B^0+\frac{u_2}{\sqrt n},
\theta^0+\frac{u_3}{n^{3/2}}
\right)
-
Q_n(A^0,B^0,\theta^0).
\]

\paragraph{Knight identity.}

Using Knight's identity \citep{knight1998limiting},
\[
|x-h|-|x|
=
-h\,\operatorname{sgn}(x)
+
2\int_0^h
\left(
\mathbf 1\{x\le s\}
-
\mathbf 1\{x\le0\}
\right)\,ds,
\]
with $x=\varepsilon_t$ and
\[
h_t
=
A^0\sin(\theta^0 t)
+
B^0\cos(\theta^0 t)
-
A\sin(\theta t)
-
B\cos(\theta t),
\]
we obtain
\begin{align*}
n\Bigl(
Q_n(A,B,\theta)
-
Q_n(A^0,B^0,\theta^0)
\Bigr)
&=
-
\sum_{t=1}^n
h_t\operatorname{sgn}(\varepsilon_t)
\\
&\quad
+
2
\sum_{t=1}^n
\int_0^{h_t}
\left(
\mathbf 1\{\varepsilon_t\le s\}
-
\mathbf 1\{\varepsilon_t\le0\}
\right)\,ds.
\end{align*}

\paragraph{Linearization of $h_t$.}

Using a first-order Taylor expansion,
\begin{align*}
h_t
&=
A^0\sin(\theta^0 t)
+
B^0\cos(\theta^0 t)
\\
&\quad
-
\left(
A^0+\frac{u_1}{\sqrt n}
\right)
\sin\!\left(
\theta^0 t
+
\frac{u_3}{n^{3/2}}t
\right)
\\
&\quad
-
\left(
B^0+\frac{u_2}{\sqrt n}
\right)
\cos\!\left(
\theta^0 t
+
\frac{u_3}{n^{3/2}}t
\right)
\\
&=
-
\frac{u_1}{\sqrt n}
\sin(\theta^0 t)
-
\frac{u_2}{\sqrt n}
\cos(\theta^0 t)
\\
&\quad
-
\frac{u_3}{n^{3/2}}
\Bigl(
A^0 t\cos(\theta^0 t)
-
B^0 t\sin(\theta^0 t)
\Bigr)
+
o(n^{-1/2}).
\end{align*}

\paragraph{Stochastic term.}

Substituting into the first sum,
\[
\sum_{t=1}^n
h_t\operatorname{sgn}(\varepsilon_t)
=
u^\top W_n
+
o_p(1),
\]
where
\[
W_n
=
\begin{pmatrix}
\dfrac1{\sqrt n}
\sum_{t=1}^n
\operatorname{sgn}(\varepsilon_t)
\sin(\theta^0 t)
\\[2ex]
\dfrac1{\sqrt n}
\sum_{t=1}^n
\operatorname{sgn}(\varepsilon_t)
\cos(\theta^0 t)
\\[2ex]
\dfrac1{n^{3/2}}
\sum_{t=1}^n
\operatorname{sgn}(\varepsilon_t)
\Bigl(
A^0 t\cos(\theta^0 t)
-
B^0 t\sin(\theta^0 t)
\Bigr)
\end{pmatrix}.
\]

By the multivariate central limit theorem,
\[
W_n
\xrightarrow{d}
\mathcal N(0,\Sigma),
\]
where
\[
\Sigma
=
\begin{pmatrix}
\frac12 & 0 & 0 \\
0 & \frac12 & 0 \\
0 & 0 &
\dfrac{(A^0)^2+(B^0)^2}{6}
\end{pmatrix},
\]
using standard trigonometric averaging arguments.

\paragraph{Deterministic term.}

For the second term,
\[
\mathbb E\!\left[
2
\int_0^{h_t}
\left(
\mathbf 1\{\varepsilon_t\le s\}
-
\mathbf 1\{\varepsilon_t\le0\}
\right)\,ds
\right]
=
f(0)h_t^2
+
o(h_t^2).
\]

Summing over $t$ yields
\[
u^\top J u,
\]
where
\[
J
=
f(0)
\begin{pmatrix}
\frac12 & 0 & 0 \\
0 & \frac12 & 0 \\
0 & 0 &
\dfrac{(A^0)^2+(B^0)^2}{6}
\end{pmatrix},
\]
which is positive definite since $f(0)>0$.

\paragraph{Limiting objective.}

Combining terms,
\[
n
\Bigl(
Q_n(A,B,\theta)
-
Q_n(A^0,B^0,\theta^0)
\Bigr)
=
-
u^\top W_n
+
u^\top J u
+
o_p(1).
\]

The limiting random objective is
\[
L(u)
=
-
u^\top Z
+
u^\top J u,
\qquad
Z\sim\mathcal N(0,\Sigma),
\]
which is minimized at
\[
u^\ast
=
\frac12
J^{-1}Z.
\]

\paragraph{Asymptotic distribution.}

Therefore,
\[
\begin{pmatrix}
\sqrt n(\hat A_n-A^0)
\\
\sqrt n(\hat B_n-B^0)
\\
n^{3/2}(\hat\theta_n-\theta^0)
\end{pmatrix}
\xrightarrow{d}
\mathcal N\!\left(
0,\,
\frac14
J^{-1}\Sigma J^{-1}
\right)
=
\mathcal N\!\left(
0,\,
\frac{1}{4f(0)^2}
\begin{pmatrix}
2 & 0 & 0 \\
0 & 2 & 0 \\
0 & 0 &
\dfrac{6}{(A^0)^2+(B^0)^2}
\end{pmatrix}
\right).
\]
\end{proof}

\end{document}